\documentclass[journal,twocolumn]{IEEEtran}
\usepackage{amssymb}
\usepackage{amsfonts}
\usepackage{amsthm}
\usepackage{amsmath}
\usepackage{algorithm}
\usepackage{algorithmic}
\usepackage{subeqnarray}
\usepackage{cases}
\usepackage{mathrsfs}
\usepackage{algorithm}
\usepackage{algorithmic}
\usepackage{balance}
\newtheorem{prop}{Proposition}

\usepackage{graphicx,subfigure,amsmath,amssymb,cite}

\usepackage[dvips]{color}
\usepackage{float}
\ifCLASSINFOpdf
\else
\fi
\begin{document}

\title{\huge Low-Complexity Leakage-based Secure Precise Wireless Transmission with Hybrid Beamforming
}

\author{Tong Shen,~~Yan Lin, \emph{Member,~IEEE},~~Jun Zou,~~Yongpeng Wu, \emph{Senior Member, IEEE},~~Feng Shu, \emph{Member,~IEEE},~~Jiangzhou Wang, \emph{ Fellow, IEEE}

\thanks{Tong Shen,~Yan Lin,~Jun Zou,~and Feng Shu are with the School of Electronic and Optical Engineering, Nanjing University of Science and Technology, 210094, CHINA. (Email: shentong0107@163.com, yanlin@njust.edu.cn, shufeng0101@163.com). }
\thanks{Y. Wu is with the Department of Electronic Engineering, Shanghai Jiao Tong University, Shanghai 200240, China (E-mail: yongpeng.wu@sjtu.edu.cn).}
\thanks{Jiangzhou Wang is with the School of Engineering and Digital Arts, University of Kent, Canterbury CT2 7NT, U.K. Email: j.z.wang@kent.ac.uk.}}

\maketitle

\begin{abstract}
 In conventional secure precise wireless transmission (SPWT), fully digital beamforming (FDB) has a high secrecy performance in transmit antenna system, but results in a huge RF-chain circuit budget for medium-scale and large-scale systems. To reduce the complexity, this letter considers a hybrid digital and analog (HDA) structure with random frequency mapped into the RF-chains to achieve SPWT. Then, a hybrid SPWT scheme based on maximizing signal-to-leakage-and-noise ratio (SLNR) and artificial-noise-to-leakage-and-noise ratio (ANLNR) (M-SLNR-ANLNR) is proposed. Compared to the FDB scheme, the proposed scheme reduces the circuit budget with low computational complexity and comparable secrecy performance.
\end{abstract}
\section{Introduction}
Secure precise wireless transmission (SPWT) is a promising physical layer security technology \cite{Wang2012Distributed,Zhao2016Anti,ChenX2017,Zou2016Relay} which transmits confidential signal precisely to a determined direction and distance, while {blue}makes confidential signal seriously distorted out of the specified range. Based upon the technique of direction modulation (DM) \cite{Wu2019,Daly2009Directional,Yan2016Artificial}, SPWT can only guarantee the security in a specific direction. As a further advance, the authors of \cite{antonik2009investigation} and \cite{sammartino2013frequency} proposed a linear frequency diverse array (LFDA) method to achieve SPWT. However, the direction angle and distance achieved by LFDA may be coupled, which means that there may exist multiple directions and distances receiving the same confidential message as the desired users received. To address this problem, the authors of \cite{liu2016range ,Hu2017SPWT,shen2019} proposed a random frequency diverse array (RFDA) scheme, in which the frequency is randomly allocated on each transmit antenna and the correlation direction and distance are decoupled.

Nevertheless, all the above SPWT schemes are based on full digital beamforming (FDB) structure, and as the number of antennas tends to be in a medium or large scale, the computational complexity of digital beamforming is significantly increased. Furthermore, the circuit cost of the FDB structure is too expensive for commercial applications. Thus, a hybrid digital and analog (HDA) beamforming structure is emerging which is less expensive and has low complexity but with comparable performance. In \cite{zhang2005variable}, Zhang \textit{et al.} firstly proposed a hybrid precoding algorithm to make a balance between circuit complexity and system performance. Shu \textit{et al.} of \cite{Qin2018} proposed a low-complexity hybrid structure beamforming in massive MIMO receive array.

Against the above researches, this paper focuses on the design of HDA beamforming. Our main contributions are summarized as follows:
\begin{enumerate}
 \item  We propose a novel hybrid maximizing signal-to-leakage-and-noise ratio (SLNR) and artificial-noise-to-leakage-and-noise ratio (ANLNR) (M-SLNR-ANLNR) scheme by digital beamforming (DB) vector, analog beamforming (AB) vector and artificial-noise (AN) vector which achieves the suboptimal secrecy rate (SR). Meanwhile, this algorithm has addressed the modulo-one constrained problem in analog beamforming.

 \item  The convergency and complexity have been analysed and compared to the conventional full digital scheme.Our proposed scheme can achieve secrecy performance close to that of the conventional FDB schemes, and secrecy performance better than that of the hybrid scheme based on equal-amplitude (EA) beamforming.
\end{enumerate}


The remainder of this letter is organized as follows. In Section II, we propose a hybrid SPWT system model. Then, a low complexity hybrid M-SLNR-ANLNR scheme is presented in Section III. The performance of the proposed method is evaluated in Section IV. Finally, the conclusions are drawn in Section V.

$Notations$: throughout the paper, matrices, vectors, and scalars are denoted by letters of bold upper case, bold lower case, and lower case, respectively. Signs $(\cdot)^T$, $(\cdot)^H$, and $(\cdot)^{-1}$ denote matrix transpose, conjugate transpose, and Moore-Penrose inverse, respectively. The symbol $\mathbf{I}_K$ denotes the $K\times K$ identity matrix.
\section{System Model}

\begin{figure}[t]
\centering
\includegraphics[width=0.40\textwidth]{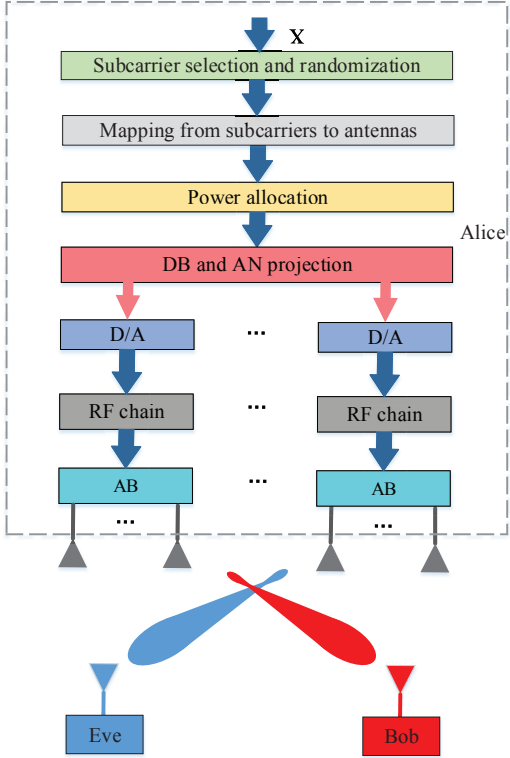}\\
\caption{A SPWT system with HDA beamforming.}\label{system_model.eps}
\end{figure}

 Considering a SPWT system which consists of a uniformly-spaced linear array (ULA) transmitter Alice, a legitimate user Bob and an eavesdropper Eve. It is assumed that Alice has $N$ antenna elements which is divided into $K$ sub-arrays, and each sub-array is composed of $M$ antenna elements, i.e., $N=KM$. Both Bob and Eve are equipped with a single antenna and located on $(\theta_B,R_B)$ and $(\theta_E,R_E)$, respectively. In our assumption, their positions are known by Alice. Fig. \ref{system_model.eps} provides an example of SPWT system with HDA beamforming.

With DB, the baseband precoded  $\mathbf{s}_{d}$ is given by
\begin{align}\label{sd}
\mathbf{s}_{d}=\sqrt {\alpha P} \mathbf{v}_dx+\sqrt {(1-\alpha)P}\mathbf{w},
\end{align}
where
\begin{align}\label{vd}
\mathbf{v}_d=[v_{d_1}, v_{d_2}, \cdots, v_{d_K}]^T
\end{align}
is the DB vector subjecting to $\mathbf{v}_d^H\mathbf{v}_d=1$, and $\mathbf{w}$ denotes the AN beamforming vector with the constraint $\mathrm{tr}\{\mathbf{w}^H\mathbf{w}\}=1$.
Moreover, $x$ is a complex digital modulation symbol with $\mathbb{E}\{x^*x\}=1$, $P$ denotes the total transmit power and $0\leq\alpha\leq1$ is the power allocation factor.
Then, after AB operation, the transmit signal vector becomes
\begin{align}\label{sa}
\mathbf{s}_{a}=V_A\mathbf{s}_{d}
=\sqrt {\alpha P} \mathbf{V}_A\mathbf{v}_dx+\sqrt {(1-\alpha)P}\mathbf{V}_A\mathbf{w},
\end{align}
where $\mathbf{V}_A$ denotes the block diagonal AB matrix which can be expressed as
\begin{align}\label{VA}
\mathbf{V}_A=\left[ {\begin{array}{*{20}{c}}
{{\mathbf{v}_{{a_1}}}}&0& \cdots &0\\
0&{{\mathbf{v}_{{a_2}}}}& \cdots &0\\
 \vdots & \vdots & \ddots & \vdots \\
0&0& \cdots &{{\mathbf{v}_{{a_K}}}}
\end{array}} \right].
\end{align}
Herein,
\begin{align}\label{vk}
\mathbf{v}_{{a_k}}=\frac{1}{\sqrt{M}}[e^{j\varphi_{k,1}},e^{j\varphi_{k,2}},\cdots,e^{j\varphi_{k,M}}]^T,
\end{align}
is the AB vector of $k$-th sub-array, where ${\varphi_{k,m}}, (m=1,2,\cdots,M)$ is the phase shifting.

For a  ULA transmitter Alice having $N$ elements, $M$ transmit antennas of each sub-array are combined with the same RF chain. Thus with the central frequency $f_c$ and the sub-carrier \cite{Zhu2009Chunk,Zhu2012Chunk} bandwidth $\Delta f$, the frequency allocated to the $k$-th sub-array is given by
\begin{align}\label{fn}
f_k=f_c+\eta_k\Delta f,~~~~k=1,2,\cdots,K,
\end{align}
where $\eta_k$ is the sub-carrier index of the $k$-th sub-array. Suppose that there are totally $N_f$ sub-carriers in our system, and $\eta_k\in S_{\eta}$, $\left(\eta_i\neq \eta_j, j\neq i\right)$, where $S_{\eta}=\{1,2,\cdots,N_f\}$ is the sub-carrier index set.
Moreover, we assume that the total bandwidth $B=N_f\Delta f\ll f_c$, and $\eta_k$ is randomly selected from $S_{\eta}$. We set the first element of the array as the reference antenna.

According to \cite{sammartino2013frequency} and \cite{Shu2018SPWT}, the normalized steering vector for a specific location $(\theta,R)$ relative to the ULA is given by
\begin{align}\label{h}
\mathbf{h}(\theta,R)=\frac{1}{\sqrt{N}}&[e^{j\Psi_{1,1}},\cdots,e^{j\Psi_{k,m}},\cdots,e^{j\Psi_{K,M}}]^T,
\end{align}
where
\begin{align}\label{Phi}
&\Psi_{k,m}=\frac{2\pi}{c}(-b_{k,m} f_c d \cos\theta+\eta_{k} \Delta f R),
\end{align}
denotes the phase shift of the $k$-th sub-array and $m$-th element relative to the reference antenna,  $d$ denotes the element spacing of the ULA at Alice, and $b_{k,m}=\left(k-1\right)M+m-1$.

Then, for simplicity, we define $\mathbf{h}_B$ and $\mathbf{h}_E$ as the abbreviations of $\mathbf{h}(\theta_B,R_B)$ and $\mathbf{h}(\theta_E,R_E)$, respectively. Considering a communication system with line of sight (LoS) channel, and far-field communication, the received signal at Bob can be expressed as
\begin{align}\label{yb}
y_B
&=\sqrt {\alpha P}\mathbf{h}^H_B \mathbf{V}_A\mathbf{v}_dx+\sqrt {(1-\alpha)P}\mathbf{h}^H_B\mathbf{V}_A\mathbf{w}+n_B,
\end{align}
where $n_B$ denotes the received additive white Gaussian noise (AWGN) at Bob with the distribution $\mathcal{CN}(0,\sigma^2)$.
Similarly, the received signal at Eve is given by
\begin{align}\label{ye}
y_E&=\sqrt {\alpha P}\mathbf{h}^H_E \mathbf{V}_A\mathbf{v}_dx+\sqrt {(1-\alpha)P}\mathbf{h}^H_E\mathbf{V}_A\mathbf{w}+n_E,
\end{align}
where $n_E$ denotes the received AWGN at Eve with the distribution $\mathcal{CN}(0,\sigma^2)$.
Hence, SR can be represented as
\begin{align}\label{SR}
\mathrm{SR}&=\log_{2}(1+ \frac{{\alpha P}|\mathbf{h}^H_B \mathbf{V}_A\mathbf{v}_d|^2}{{(1-\alpha)P}|\mathbf{h}^H_B\mathbf{V}_A\mathbf{w}|+n^2_B})\nonumber\\
&~~~-\log_{2}(1+ \frac{{\alpha P}|\mathbf{h}^H_E \mathbf{V}_A\mathbf{v}_d|^2}{{(1-\alpha)P}|\mathbf{h}^H_E\mathbf{V}_A\mathbf{w}|+n^2_E}).
\end{align}


\section{The Proposed Low-Complexity Leakage-based Hybrid Beamforming}
In order to reduce the computational complexity of SPWT scheme, we propose a M-SLNR-ANLNR hybrid beamforming scheme. Due to the fact that the AB vector $V_A$ is highly coupled with $\mathbf{v}_d$ and $\mathbf{w}$ in the SR expression, we choose to optimize the SLNR and ANLNR. This is owing to the fact that the problem of maximizing SR can be approximated as maximizing the product of SLNR and ANLNR in the medium and large SNR region. In the following, we firstly give the suboptimal beamforming vector by maximizing SLNR and then obtain the AN vector by maximizing ANLNR, respectively.
\subsection{Proposed Scheme}
In terms of HDA beamforming, DB vector $\mathbf{v}_d$, the AB matrix $\mathbf{V}_A$ and AN beamforming vector $\mathbf{w}$ are our optimization variables. To reduce the computational complexity, we let
\begin{align}\label{Vaij}
\mathbf{v}_{a_i}=\mathbf{v}_{a_j}=\mathbf{v}_a,(i,j=1,2,\cdots,K),
\end{align}
and
\begin{align}\label{va}
\mathbf{v}_{{a}}=\frac{1}{\sqrt{M}}[e^{j\varphi_{1}},e^{j\varphi_{2}},\cdots,e^{j\varphi_{M}}]^T.
\end{align}
Then, for the desired user, the SLNR can be expressed as
\begin{align}\label{SLNRB}
\mathrm{SLNR}=\frac{\alpha P |\mathbf{h}_B^H \mathbf{V}_A \mathbf{v}_d|^2}{\alpha P |\mathbf{h}_E^H \mathbf{V}_A \mathbf{v}_d|^2+\sigma^2},
\end{align}
and hence the optimization problem is written as
\begin{align}\label{max}
&\mathop {{\rm{maximize}}}\limits_{\mathbf{V}_A,\mathbf{v}_d}  ~\mathrm{SLNR}_B(\mathbf{V}_A,\mathbf{v}_d),\nonumber\\
&\mathrm{s.t.}~{\mathbf{v}_d^H}\mathbf{v}_d =1 .
\end{align}


It can be clearly seen from (\ref{max}) that it is difficult to solve $\mathbf{V}_A$ and $\mathbf{v}_d$ simultaneously because they are coupled in (\ref{SLNRB}) and each element in $V_A$ has the same amplitude. In the following, we will first solve $\mathbf{V}_A$ with fixed $\mathbf{v}_d$, then $\mathbf{v}_d$ with solved $\mathbf{V}_A$, last $\mathbf{w}$ with solved $\mathbf{V}_A$ and $\mathbf{v}_d$.

\textit{$\mathbf{1)}$ Solving for $\mathbf{V}_A$: With any given initial value of $\mathbf{v}_d[0]$ which satisfies $\mathbf{v}^H_d[0]\mathbf{v}_d[0]=1$, the optimization problem (\ref{max}) can be rewritten as
\begin{align}\label{max1}
&\mathop {{\rm{maximize}}}\limits_{\varphi_{1},\varphi_{2},\cdots,\varphi_{M}}  ~\mathrm{SLNR}_B(\varphi_{1},\varphi_{2},\cdots,\varphi_{M}).
\end{align}
\begin{prop}\label{prop1}
Given an initial $\mathbf{v}_a[0]$, and for each $\varphi_{m}$,  the optimal $\varphi_{m}$ can be obtained by gradient descent method when other $\varphi_{i,(i\neq m)}$ are fixed.
\begin{align}\label{phim}
\varphi_m=\arcsin \frac{2R_{E,m}I_{B,m}-2R_{B,m}I_{E,m}}{\sqrt{y_1^2+y_2^2}}-\varphi_y,
\end{align}
where the denotations of $R_{E,m}$, $I_{B,m}$, $R_{B,m}$, $I_{E,m}$, $y_1$, and $y_2$ are provided in Appendix A.
\end{prop}
\begin{proof}\label{proof}
The proof is shown in Appendix A.
\end{proof}
Based upon Proposition 1, by a number of iterations, we can obtain the optimal solution of $\mathbf{v}_a^*$. }

\textit{$\mathbf{2)}$ Solving for $\mathbf{v}_d$}: Next, we solve the optimal value of $\mathbf{v}_d$. As $\mathbf{v}_a^*$ is obtained, by substituting $\mathbf{v}_a^*$ into (\ref{SLNRB}), the expression of SLNR can be rewritten as
\begin{align}\label{SLNRB1}
{\rm{SLNR}} = \frac{{{\bf{v}}_d^H[\alpha P{\bf{V}}_A^H{{\bf{h}}_B}{\bf{h}}_B^H{{\bf{V}}_A}]{{\bf{v}}_d}}}{{{\bf{v}}_d^H[ \alpha P{\bf{V}}_A^H{{\bf{h}}_E}{\bf{h}}_E^H{{\bf{V}}_A} + {\sigma ^2}{\mathbf{I}_K}]{{\bf{v}}_d}}}.
\end{align}
According to the generalized Rayleigh-Ritz theorem \cite{Horn1985Matrix}, the optimal $\mathbf{v}_d$ is the normalized eigenvector corresponding to the largest eigenvalue of
\begin{align}\label{eig}
{[ \alpha P{\bf{V}}_A^H{{\bf{h}}_E}{\bf{h}}_E^H{{\bf{V}}_A} + {\sigma ^2}{\mathbf{I}_K}]^{ - 1}}(\alpha P{\bf{V}}_A^H{{\bf{h}}_B}{\bf{h}}_B^H{{\bf{V}}_A}).
\end{align}

\begin{algorithm}[t]
\caption{Algorithm for solving problem (\ref{max})}
\label{alg:A}
\begin{algorithmic}[1]
\STATE {Initialize $\mathbf{v}_d[0]$ and $\mathbf{v}_a[0]$ randomly that is feasible to (\ref{va}), and ${\mathbf{v}_d[0]}^H\mathbf{v}_d[0]=1$, respectively;}
\STATE {Set $n=0$;}
\REPEAT
\STATE Set $m=0$;
\REPEAT
\STATE {Substituting $\mathbf{v}_d[n]$ and $\mathbf{v}_a[m]$ into (\ref{max1}) yields $\mathbf{v}_a[m+1]$ by gradient descent method;}
\STATE {Updating $m=m+1$;}
\UNTIL{Convergence}
\STATE Substituting $\mathbf{v}_a[m]$ into (\ref{eig}) yields the optimal solution $\mathbf{v}_d[n+1]$;
\STATE Updating $n=n+1$;
\UNTIL{Convergence}
\STATE {The final optimal solution $\mathbf{v}_a^*=\mathbf{v}_a[m]$ and $\mathbf{v}_d^*=\mathbf{v}_d[n]$.}
\STATE {Obtaining $\mathbf{w}$ according to (\ref{eigw}).}
\end{algorithmic}
\end{algorithm}


\textit{$\mathbf{3)}$ Solving for $\mathbf{w}$}: Finally, we design the AN vector by maximizing the ANLNR, which is defined as
\begin{align}\label{ANLNR}
{\rm{ANLNR}} =& \frac{(1-\alpha) P |\mathbf{h}_E^H \mathbf{V}_A \mathbf{w}|^2}{(1-\alpha) P |\mathbf{h}_B^H \mathbf{V}_A \mathbf{w}|^2+\sigma^2},\nonumber\\
=&\frac{{{\bf{w}}^H[(1-\alpha) P{\bf{V}}_A^H{{\bf{h}}_E}{\bf{h}}_E^H{{\bf{V}}_A}]{{\bf{w}}}}}{{{\bf{w}}^H[(1- \alpha) P{\bf{V}}_A^H{{\bf{h}}_B}{\bf{h}}_B^H{{\bf{V}}_A} + {\sigma ^2}{\mathbf{I}_K}]{{\bf{w}}}}}.
\end{align}
Then, the optimal $\mathbf{w}$ is the normalized eigenvector corresponding to the largest eigenvalue of
\begin{align}\label{eigw}
{[ (1-\alpha) P{\bf{V}}_A^H{{\bf{h}}_B}{\bf{h}}_B^H{{\bf{V}}_A} + {\sigma ^2}{I_K}]^{ - 1}}((1-\alpha) P{\bf{V}}_A^H{{\bf{h}}_E}{\bf{h}}_E^H{{\bf{V}}_A}).
\end{align}
The complete iterative procedure of solving problem (\ref{max}) is summarized in Algorithm \ref{alg:A}.

\subsection{Convergence and Complexity Analysis}
\textit{Convergence:} From Step 5 to Step 8 in Algorithm 1, we can obtain $\mathbf{v}_a[m+1]$ with Appendix A, thus we have $\mathrm{SLNR}(\mathbf{v}_a[m+1],v_d[n])\geq \mathrm{SLNR}(\mathbf{v}_a[m],v_d[n])$. Then, we obtain $v_d[n+1]$ with (\ref{eig}), and we have $\mathrm{SLNR}(\mathbf{v}_a[m+1],v_d[n+1])\geq \mathrm{SLNR}(\mathbf{v}_a[m+1],v_d[n])$. As a result, the value of SLNR increases after each iteration, due to $\mathrm{SLNR}\leq \frac{\alpha P}{\sigma^2}$. Therefore, Algorithm 1 converges.

\textit{Complexity Analysis:} According to (\ref{eig}) and Appendix A, the computational complexity of DB is $\mathcal{O}(3K^3+2K(2N+K))$ in terms of floating-point operations (FLOPs), while the computation complexity of AB is $\mathcal{O}(K^2+KN)$ FLOPs. Additionally, as shown in (\ref{eigw}), computing AN vector $\mathbf{w}$ requires $\mathcal{O}(3K^3+2K(2N+K))$ FLOPs. Thus, the total computational complexity of this hybrid beamforming scheme is $\mathcal{O}(T*(3K^3+3K^2+5KN)+3K^3+2K(2N+K))$ FLOPs, where $T$ is the iteration times.

By comparison, for the conventional FDB scheme \cite{Shu2018SPWT}, the computational complexity of DB $\mathbf{v}_d$ is $\mathcal{O}(2N^3+2N^2)$ FLOPs. In addition, the computational complexity for AN vector $\mathbf{w}$ is $\mathcal{O}(N^3+2N^2)$ FLOPs, then we finally have the total computational complexity of the full DB scheme as $\mathcal{O}(3N^3+4N^2)$ FLOPs. It is noteworthy that the complexity of conventional FDB scheme is generally higher than that of proposed scheme, namely $N>K$. Consequently, our proposed hybrid beamforming scheme has a lower complexity than the FDB scheme.
\section{simulation results and analysis}
To evaluate the performance of our methods, the parameters and specifications in our simulation are used as follows. By default, we set $f_c=3$GHz, $B=20$MHz, $N_f=1024$, $d=c/2f_c$, $K=32$, and $M=4$. We assume $\sigma^2_B=\sigma^2_E$, and the location of Bob and of Eve are set in $(45^\circ,600m)$, and $(120^\circ,300m)$, respectively.
\begin{figure}[t]
\centering
\includegraphics[width=0.48\textwidth]{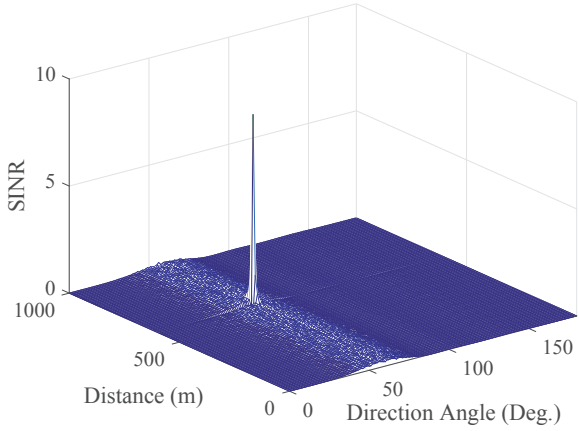}\\
\caption{3-D surface of SINR versus direction angle and distance}\label{SINR}
\end{figure}

\begin{figure}[t]
\centering
\includegraphics[width=0.48\textwidth]{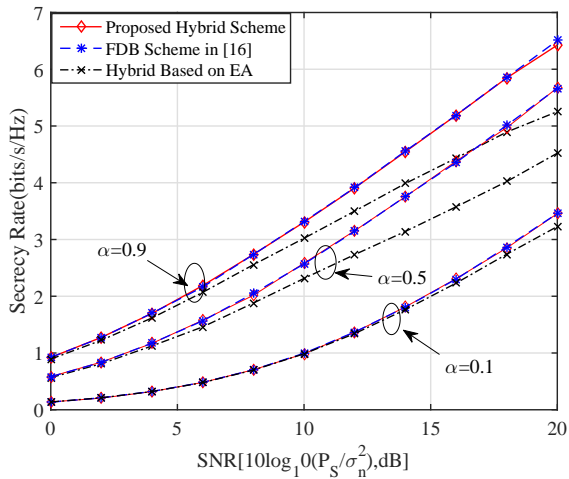}\\
\caption{Secrecy rate versus SNR with Bob and Eve located in different direction angle.}\label{SR}
\end{figure}

Fig. \ref{SINR} plots the 3-D surface of SINR versus the direction angle and the distance. It can be seen that that only a main energy peak appears at the position of Bob, which indicates our proposed scheme can achieve SPWT. Moreover, the size of the main peak depends on both the number of antennas $N$ and the total bandwidth $B$, and it becomes larger as $N$ and $B$ increase.

Fig. \ref{SR} compares the secrecy rate of the proposed hybrid M-SLNR-ANLNR scheme, the FDB scheme in \cite{Shu2018SPWT} and the hybrid scheme with equal-amplitude (EA) beamforming versus the SNR. The EA scheme designs beamforming vector by maximizing SINR at Bob as $\mathbf{h}^H_B\mathbf{V}_a\mathbf{v}_d=1$ and $\mathbf{h}^H_B\mathbf{V}_a\mathbf{w}=0$, thus $\mathrm{SINR}_B=\frac{1}{\sigma^2}$, but this scheme does not consider the SINR at Eve which may also be large. The simulation results show that our proposed hybrid scheme performs almost the same as the FDB scheme. And there exists a clear gap between proposed hybrid scheme and hybrid EA scheme, which tends to be large when $\alpha=0.5$ and $\alpha=0.9$.


\begin{figure}[t]
\centering
\includegraphics[width=0.45\textwidth]{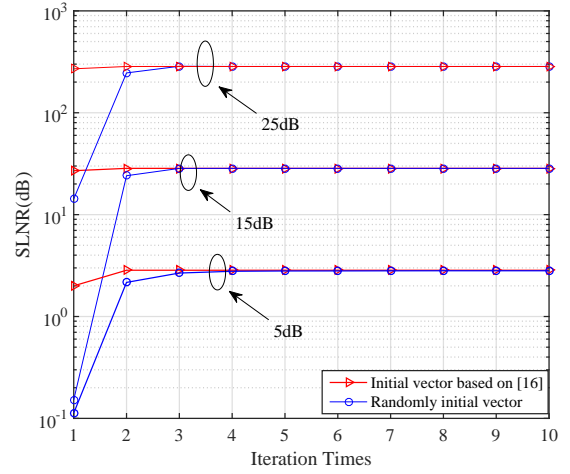}\\
\caption{SLNR versus iteration times}\label{IT}
\end{figure}

\begin{figure}[t]
\centering
\includegraphics[width=0.48\textwidth]{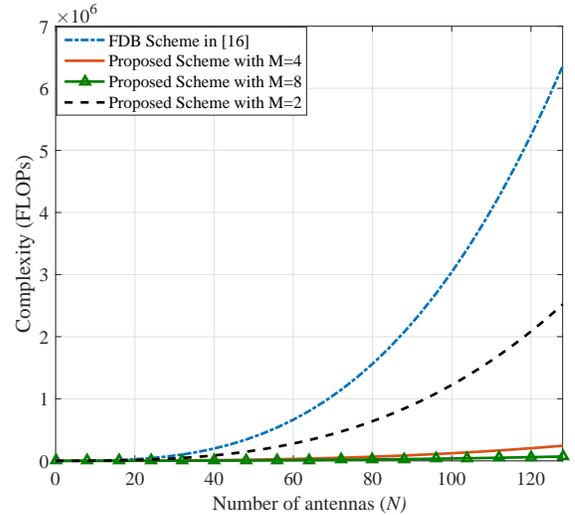}\\
\caption{Complexity versus the total number of antennas.}\label{complexity}
\end{figure}
Fig. \ref{IT} illustrates the SLNR versus iteration times. It shows that, the SLNR of our proposed scheme converges within three iterations, and with the initial vector obtained in \cite{Shu2018SPWT}, it converges within only two iterations. Fig. \ref{complexity} illustrates the curve of the complexity versus the total number of antennas. From this figure, it is seen that the complexity of our proposed scheme grows much slower than that of the conventional FDB scheme, as the total number of antennas increases, in particular with in a medium or large scale. Notably, the complexity of the proposed scheme increases as $M$ decreases, since $K$ increases with $N$, and the complexity of the proposed scheme depends on $K$ and $N$.

In summary, the proposed hybrid scheme can achieve SPWT with low complexity. Therefore, although the FDB scheme has a better performance than the proposed hybrid scheme due to the AB vector having a equal-modular constraint, our proposed hybrid scheme has the close performance to the FDB scheme, as well as has a significantly superior performance to the EA hybrid scheme.
%
%
%

\section{Conclusion}
In this letter, we proposed a M-SLNR-ANLNR SPWT scheme based on hybrid beamforming in order to reduce the circuit budget and the complexity in the medium-scale and large-scale SPWT systems. To solve the formulation problem, we presented a pair of beamforming and AN vectors by maximizing SLNR and ANLNR, respectively. The simulation results have shown that our proposed scheme has almost the same secrecy performance as and much lower complexity than the conventional FDB scheme. Finally, the proposed scheme is superior to the hybrid EA scheme in term of SR.

\begin{appendices}
\section{DERIVATION OF $\varphi_{m}$}
By substituting Eqs. (\ref{VA}), (\ref{Vaij}) and (\ref{va}) into $\mathbf{h}_B^H \mathbf{V}_A \mathbf{v}_d$ and $\mathbf{h}_E^H \mathbf{V}_A \mathbf{v}_d$, we have
\begin{align}\label{hbvavd}
\mathbf{h}_B^H \mathbf{V}_A \mathbf{v}_d&=\sqrt{\frac{1}{NM}}\sum\limits_{k = 1}^K {\sum\limits_{m = 1}^M {{e^{ - j{\Psi _{k,m}}}}{e^{j{\varphi _m}}}} {v_{{d_k}}}}\nonumber\\
 &= {e^{j{\varphi _m}}}{A_{{B_m}}} + {B_{{B_m}}},
\end{align}
where
\begin{align}\label{ABm}
&{A_{{B_m}}}=\sqrt{\frac{1}{NM}}{\sum\limits_{k = 1}^K {{e^{ - j{\Psi _{k,m}}}}{v_{{d_k}}}} },\\
&{B_{{B_m}}}=\sqrt{\frac{1}{NM}}\sum\limits_{i = 1,i \ne m}^M {{e^{j{\varphi _m}}}\sum\limits_{k = 1}^K {{e^{ - j{\Psi _{k,m}}}}{v_{{d_k}}}} } .
\end{align}
\end{appendices}
Then, we have
\begin{align}\label{hbvavd2}
|\mathbf{h}_B^H \mathbf{V}_A \mathbf{v}_d|^2&=({e^{j{\varphi _m}}}{A_{{B_m}}} + {B_{{B_m}}})^H({e^{j{\varphi _m}}}{A_{{B_m}}} + {B_{{B_m}}})\nonumber\\
&=A_{{B_m}}^HA_{{B_m}}+B_{{B_m}}^HB_{{B_m}}\nonumber\\
&~~~~+2(R_{B,m}\cos \varphi_m-I_{B,m}\sin \varphi_m),
\end{align}
with
\begin{align}\label{RBm1}
&R_{B,m}=Re(A_{{B_m}}B_{{B_m}}^H),
\end{align}
\begin{align}\label{RBm2}
&I_{B,m}=Im(A_{{B_m}}B_{{B_m}}^H).
\end{align}
where $Re(\cdot)$ and $Im(\cdot)$ denote the real part and the imaginary part, respectively.

Similarly, we can obtain $|\mathbf{h}_E^H \mathbf{V}_A \mathbf{v}_d|^2$, $A_{{E_m}}$, $B_{{E_m}}$, $R_{E,m}$ and $I_{E,m}$. Thus, the SLNR can be expressed as
\begin{align}\label{SLNR2}
\mathrm{SLNR}_B=\frac {X_{B,m}+2\alpha P R_{B,m} \cos \varphi_m - 2\alpha P I_{B,m}\sin \varphi_m} {X_{E,m}+2\alpha P R_{E,m} \cos \varphi_m - 2\alpha P I_{E,m}\sin \varphi_m}
\end{align}
where
\begin{align}\label{XB}
&X_{B,m}=\alpha P(A_{{B_m}}^HA_{{B_m}}+B_{{B_m}}^HB_{{B_m}}),\\
&X_{E,m}=\alpha P(A_{{E_m}}^HA_{{E_m}}+B_{{E_m}}^HB_{{E_m}})+\sigma^2.
\end{align}
Therefore, in order to obtain the optimal $\varphi_m$, we compute the derivation of $\mathrm{SLNR}_B$ with respect to $\varphi_m$,
\begin{align}\label{deri}
&\frac{\partial(\mathrm{SLNR}_B)}{\partial (\varphi_m)}=\nonumber\\
&\frac {2y_1\sin \varphi_m+2y_2\cos \varphi_m+4R_{B,m}I_{E,m}-4R_{E,m}I_{B,m}} {[X_E+2\alpha P R_{E,m} \cos \varphi_m - 2\alpha P I_{E,m}\sin \varphi_m]^2},
\end{align}
where
\begin{align}\label{y1}
&y_1=X_{B,m}R_{E,m}-X_{E,m}R_{B,m},
\end{align}
and
\begin{align}\label{y2}
&y_2=X_{B,m}I_{E,m}-X_{E,m}I_{B,m}.
\end{align}
Let $\frac{\partial(\mathrm{SLNR}_B)}{\partial (\varphi_m)}=0$, we obtain
\begin{align}\label{phim}
\varphi_m=\arcsin \frac{2R_{E,m}I_{B,m}-2R_{B,m}I_{E,m}}{\sqrt{y_1^2+y_2^2}}-\varphi_y,
\end{align}
where $\varphi_y$ satisfies $\sin \varphi_y=\frac{y_2}{\sqrt{y_1^2+y_2^2}}$ and $\cos \varphi_y=\frac{y_1}{\sqrt{y_1^2+y_2^2}}$.

\noindent This complete the derivation of $\varphi_m$.
\ifCLASSOPTIONcaptionsoff
  \newpage
\fi
\bibliographystyle{IEEEtran}

\bibliography{IEEEfull,reference}

\end{document}